\documentclass[12pt,english]{amsart}
\usepackage{amsmath,amssymb,amsthm,latexsym}
\usepackage[dvipdf]{graphicx}
\usepackage[T1]{fontenc}
\usepackage[latin9]{inputenc}
\usepackage[letterpaper]{geometry}
\usepackage{float}
\usepackage{amsthm}
\usepackage{amstext}
\usepackage{graphicx}
\usepackage{amssymb}

\newtheorem{theorem}{Theorem}[section]

\newtheorem{corollary}[theorem]{Corollary}

\theoremstyle{definition}
\newtheorem{definition}[theorem]{Definition}

\theoremstyle{remark}
\newtheorem{remark}[theorem]{Remark}

\newcommand{\cP}{\mathcal{P}}
\newcommand{\cQ}{\mathcal{Q}}
\newcommand\cV{\mathcal{V}}
\newcommand\cE{\mathcal{E}}

\newcommand{\R}{\mathbb{R}}

\newcommand{\tG}{\widetilde{G}}
\newcommand{\ta}{\widetilde{\alpha}}

\newcommand\zer{\phi}
\newcommand\morse{\mu}

\makeatother

\usepackage{babel}

\begin{document}


\title[Stability of nodal structures and nodal count]{Stability of
  nodal structures in graph eigenfunctions and its relation to the
  nodal domain count}

\author{Gregory~Berkolaiko$^{1}$}

\address{$^{\text{1}}$Dept of Mathematics, Texas A\&M
  University, College Station, TX 77843-3368, USA}

\author{Hillel Raz$^{2}$}

\address{$^{2}$Cardiff School of Mathematics, Cardiff University,
Cardiff CF24 4AG, UK}

\author{Uzy~Smilansky$^{2,\,3}$}

\address{$^{3}$Department of Physics of Complex Systems, The Weizmann Institute of Science, Rehovot 76100, Israel}

\begin{abstract}
  The nodal domains of eigenvectors of the discrete Schr\"odinger
  operator on simple, finite and connected graphs are considered.
  Courant's well known nodal domain theorem applies in the present
  case, and sets an upper bound to the number of nodal domains of
  eigenvectors: Arranging the spectrum as a non decreasing sequence,
  and denoting by $\nu_n$ the number of nodal domains of the $n$'th
  eigenvector, Courant's theorem guarantees that the nodal deficiency
  $n-\nu_n$ is non negative. (The above applies for generic
  eigenvectors. Special care should be exercised for eigenvectors with
  vanishing components.)  The main result of the present work is that
  the nodal deficiency for generic eigenvectors equals to a Morse
  index of an energy functional whose value at its relevant critical
  points coincides with the eigenvalue. The association of the nodal
  deficiency to the stability of an energy functional at its critical
  points was recently discussed in the context of quantum graphs
  \cite{BBRS} and Dirichlet Laplacian in bounded domains in
  $\mathbb{R}^d$ \cite{BKS}. The present work adapts this result to
  the discrete case.  The definition of the energy functional in the
  discrete case requires a special setting, substantially different
  from the one used in \cite{BBRS,BKS} and it is presented here in
  detail.
\end{abstract}

\maketitle

\section{Introduction}

Courant's nodal domain theorem can be viewed as a generalization of
Sturm's oscillation theorem to Laplace-Beltrami operators in higher
dimensions. In the one-dimensional case, after ordering the spectrum
of the Sturm-Liouville operator as an increasing sequence, the
oscillation theorem guarantees that the $n$'th eigenfunction $f^{(n)}$
flips sign $\zer_n=n-1$ times in the open interval. Equivalently, the
number of nodal domains $\nu_n$ --- defined as the number of intervals
where $f^{(n)}$ is of constant sign --- equals $n$. The Sturm
oscillation theorem can be written concisely as
\begin{equation}\label{eq:sturm}
 n=\nu_n=\zer_n+1 \, .
\end{equation}
In higher dimensions, the two equalities in (\ref{eq:sturm}) have to
be modified. For $d>1$ there is no natural analogue of $\zer_n$, and
therefore the right equality has to be discarded. Courant
\cite{Courant23} showed that the left equality cannot hold in general,
and it should be replaced by a bound:
\begin{equation}
\label{eq:courant} n\ge \nu_n,
\end{equation}
where $\nu_n$ is the number of maximal connected components of the
domain on which the $n$'th eigenfunction has constant sign.  Later
studies \cite{Pleijel56,Pol_pams09} have shown that equality holds
only for finitely many eigenvectors.  These states are referred to as
Courant sharp.

Courant's nodal domain theorem was extended to the Laplace operator on
metric (quantum) graphs \cite{GSW03} and on discrete graphs
\cite{DGLS01} (see also references therein).  The latter are the
subject of the present work, and they will be discussed in detail in
the next section.

The interest in counting the nodal domains increased in the
mathematical and physical communities when it was realized that the
nodal sequence $\{ \nu_n \}_{n=1}^{\infty}$ stores metric information
about the manifold where the Laplace-Beltrami operator is defined
\cite{BGS02,BDS07,BS02,BOS08,Klawonn09,BKP07}. It was shown, in
particular, that in some cases the nodal sequences of isospectral
domains are different \cite{GSS05,BKP07,BS07,BSS06} and that in some
other examples one can uniquely reconstruct the domain geometry from
the given nodal sequence \cite{PKUS08}.

A new point of view was introduced in the pioneering article of
Helffer, Hoffmann-Ostenhof and Terracini \cite{HHT09}, where a
variational approach was used to locate the Courant sharp
eigenfunctions in the spectrum.  Helffer \emph{et. al.} investigated
the Dirichlet Laplacian in a bounded domain $\Omega \in
\mathbb{R}^{d}$ with $d \ge 2$.  Partitioning $\Omega$ arbitrarily
into $\nu$ sub-domains $\Omega_k$, they studied the lowest Dirichlet
eigenvalue $\lambda_1(\Omega_k)$ for each of the sub-domains. The
maximal value amongst the $\lambda_1(\Omega_k)$ for a given partition
$\mathcal{P}$ (denoted as $\Lambda(\nu;\mathcal{P})$) can be viewed as
the ``energy'' of the partition.  Helffer \emph{et. al.}  proved that
the partitions that minimize $\Lambda(\nu;\mathcal{P})$ coincide with
the nodal partition induced by a Courant sharp eigenfunction if and
only if the minimizing partition is bipartite.

We note that that the restriction of an eigenfunction $f^{(n)}$ to a
nodal domain $\Omega_k$ is the ground state of the Dirichlet Laplacian on
$\Omega_n$.  The corresponding ground energy $\lambda_1(\Omega_k)$ is
equal to the eigenvalue $\lambda_n(\Omega)$ on the entire domain.
Thus a partition corresponding to an eigenfunction will have a special
property: $\lambda_1(\Omega_k)$ is the same for all $k$.  A partition
with this property will be called an equipartition.

In two recent papers \cite{BBRS,BKS} the approach of Helffer
\emph{et. al.} was broadened in a substantial way.  It was shown that
the functional $\Lambda(\nu;\mathcal{P})$, when restricted to a
submanifold of equipartitions, becomes smooth.  One can then study
the critical points of $\Lambda(\nu;\mathcal{P})$, i.e.\ the points
where the variation with respect to perturbations of the partition boundaries
vanishes.  It was shown that the critical $\nu$-partitions that are
bipartite are in one-to-one correspondence with the nodal partitions
of eigenfunctions with $\nu$ nodal domains. Moreover, the Morse index
$\morse_n$ at the critical partitions equals the nodal deficiency of
the corresponding eigenstate,
\begin{equation}\label{eq:morseg}
  \morse_n=n-\nu_n \ .
\end{equation}
Thus, in the space of bipartite equipartitions, the critical points in the
landscape of $\Lambda(\nu;\mathcal{P})$ are at eigenfunctions, and
their stability (number of directions at which the critical point is a
maximum) determines the nodal deficiency. The result (\ref{eq:morseg})
was proved for the Laplacian on metric graphs (quantum graphs) in
\cite{BBRS} and was then shown to hold on domains in $\mathbb{R}^{d}$
in \cite{BKS}. The present work complements the above mentioned
articles by showing that (\ref{eq:morseg}) applies also for the
discrete Schr\"odinger operator on finite graphs.

The variational approach used in \cite{BBRS,BKS} crucially depends on
the ability to smoothly change the boundaries of the domains.  There
is no direct analogue of this for discrete graphs.  In the present
work we show how one can use local variation of the \emph{potential}
in place of the variation of the partition boundaries.  This is the
chief new element introduced in this article, and it enables us to
arrive at the main result of the present work, namely equality (\ref
{eq:morseg}) for generic eigenvectors of the graph Shr\"odinger
operator. The proof is provided in theorem \ref{thm:morse}.

There is another special feature that distinguishes graphs (metric or
discrete) from other manifolds: one can describe the nodal structures
in terms of points where the wave function changes its sign (for
discrete graphs these will correspond to changes of the sign of the
vertex eigenvector across a connecting edge).  This allows us to
reinterpret the bound on the number of nodal domains as a
generalization of Sturm theorem. Indeed, we show that the right
equality in (\ref {eq:sturm}) is replaced by upper and lower bounds on
$\phi_n$, given in equation (\ref {eq:zero_bound_mod}), Theorem
\ref{thm:nodal_bound_mod}.

The paper is organized in the following way. The next section provides
a few definitions and known facts from spectral graph theory which are
necessary for the ensuing discussion. The construction of partitions
and the analogue of boundary variations require some ground work which
is carried out in the section \ref{section:preparations}. The main
results will be formulated and proved in section \ref{section:results}.

\section{Definitions and general background}

In the present chapter we provide a few definitions and facts which
set the stage for the subsequent discussion.

A graph $G = (\mathcal{V}, \mathcal{E})$ consists of a set of vertices
$\mathcal{V} = \mathcal{V}(G)$ and a set of connecting edges
$\mathcal{G} = \mathcal{E}(G)$. We shall use: $V(G) = |\mathcal{V}(G)|$
 and $E(G) =|\mathcal{E}(G)|$.  An edge $e\in \mathcal{E}(G)$
connecting the vertices $i, j \in \mathcal{V}(G)$ will be also denoted
by $e=(i,j)=(j,i)$.  We will use the notation $i \sim j$ to indicate
that the vertices $i$ and $j$ are connected (are neighbors).  A graph
is simple if no more than a single edge connects two vertices and no
vertex is connected to itself; otherwise the graph is a multi-graph.
A graph is said to be connected if there exists a path of connected
vertices between any two vertices in $\mathcal{V}$. Unless otherwise
specified, the graphs we consider here are finite, connected and
simple.

The connectivity of a graph $G$ is summarized by the adjacency matrix
$A(G)$.
\begin{equation}\label{eq:adj}
[A(G)]_{i,j} = \left\{\begin{array}{cc} 1, & \text{if } i\sim j \\
0, & \text{otherwise}. \end{array} \right.
\end{equation}
The degree $d_i$ of the vertex $i$ is the number of its neighbors; it
can be expressed via the adjacency matrix as
\begin{equation*}
  d_i=\sum_{j\sim i}1 = \sum_j A(G)_{i,j}.
\end{equation*}
The (first) Betty number of a connected graph is defined as
\begin{equation}
  \beta(G) =   E(G) -   V(G)+1\ .
\end{equation}
This is the number of independent cycles on the graph.

A subgraph $G' \subseteq G$ is itself a graph, defined by a subset of
the vertex set $\mathcal{V}(G') \subseteq \mathcal{V}(G)$ and a subset
of the edge set $\mathcal{E}(G') \subseteq \mathcal{E}(G)$. A factor
(or spanning graph) of $G$ is a subgraph, $G' \subseteq G$, which
shares with $G$ the vertex set, $\mathcal{V}(G') = \mathcal{V}(G)$.
By removing $\beta(G)$ edges one can generate a factor which is a tree
(a spanning tree).  There exist more than one spanning tree. In what
follows we shall often construct a sequence of connected factors by
starting with $G$ and removing edges one at a time, while keeping the
resulting factor connected and ending with a spanning tree after
$\beta(G)$ steps.

A class of factors which will play a prominent role here are
$\nu$-partitions.
\begin{definition}
  \label{def:partition}
  A $\nu$-\emph{partition} of the graph $G$, denoted by $\mathcal{P}$,
  is a factor consisting of $\nu$ \emph{disjoint} subgraphs, $P_k$,
  such that for any two vertices $i$ and $j$ in the same connected
  component $P_k$, if $(i,j) \in \mathcal E(G)$ then they are also
  connected in $P_k$.  In other words, only the edges running
  \emph{between} components were removed from $G$ to form $\mathcal
  P$.  We write $\mathcal{P}=\bigcup_{k=1}^{\nu} P_k$.
\end{definition}

We associate to any partition $\mathcal{P}$ a multi-graph
$\tG(\mathcal{P})$ with $\nu$ vertices.  The vertices of
$\tG(\mathcal{P})$ are in one-to-one correspondence with the
disjoint subgraphs of $\mathcal{P}$, while the edges of $\tG(\cP)$
correspond to the edges that were removed from $G$.  A partition is
bipartite if its graph $\tG(\mathcal{P})$ is bipartite, or a
tree if $\tG(\mathcal{P})$ is a tree, \emph{etc}.  Note that if
$\tG(\mathcal{P})$ is a tree, each pair of nodal domains is connected
by at most one edge. Edges of this kind are referred to as {\it bridges}
in the graph theory literature.

\subsection{The Schr\"odinger Operator}

The discrete Schr\"odinger operator acts on the Hilbert space of real
vectors $f\in R^{V(G)}$ where the components $f_i$ are enumerated by
the vertex indices $i$. The Schr\"odinger operator is a sum of two
operators - the Laplacian and an on-site potential. The Laplacian is
usually defined as
\begin{equation}\label{eq:Lap}
  L(G) = -A(G) + D(G)
\end{equation}
where $A(G)$ is the adjacency matrix of $G$ and $D(G)$ is a diagonal
matrix with $[D(G)]_{ii} = d_i$.

The on-site potential, $Q(G)$ is a diagonal matrix
$[Q(G)]_{ii} = q_i \in \R$.  However, since both $D$ and $Q$
are diagonal matrices and we will not impose any restrictions on the
choice of site potentials $q_i$, we can ``absorb'' the degree matrix
in the potential and define the Schr\"odinger operator
\begin{equation}\label{eq:SO}
  H(G) = -A(G) + Q(G).
\end{equation}
This convention significantly simplifies the notation later.
The action of the Schr\"odinger operator (which we will also call the
Hamiltonian) on a vector $f$ is given by
\begin{equation}
  \big[H(G)f \big]_i = -\sum_{j \sim i} f_j + q_i f_i \ .
\end{equation}

The spectrum of (\ref{eq:SO}) is discrete and finite. We denote it as
\begin{equation*}
  \sigma(G) = \{ \lambda_n(G)\}_{n = 1}^{V(G)}
  \qquad\text{with}\qquad
  \lambda_1(G) < \lambda_2(G) \le \ldots \le \lambda_n(G).
\end{equation*}
The eigenvector corresponding to $\lambda_n(G)$ will be denoted as
$f^{(n)} = \{f_i^{(n)}\}_{i=1}^{V(G)}$.

\begin{remark}
  The definitions above can be generalized by associating positive
  weights $m_{i,j} =m_{j,i}$ to the connected bonds in $G$. The
  weighted adjacency matrix is defined as $[M(G)]_{i,j} =[A(G)]_{i,j}
  m_{i,j}$ .  The results derived in the present paper are valid for
  this generalized version of the Schr\"odinger operator. However to
  simplify the notation we shall present the results for the case of
  uniform unit weights.
\end{remark}

\subsection{Nodal Domains}

Let $f : G \rightarrow \mathbb{R}^{V(G)}, \, f = \{f_i\}_{i=1}^{V(G)}$
be an arbitrary real function on the graph $G$. We define a strong
nodal domain as a maximally connected subgraph of $G$ such that on all
of its vertices the components of $f$ have the same sign. Vertices
where $f$ vanishes do not belong to a strong nodal domain. A weak
positive (respectively negative) nodal domain is defined as a
maximally connected subgraph $Y$ such that for all $i \in
\mathcal{V}(Y)$, $f_i \ge 0 \, ($respectively $f_i \le 0)$.  In this
case, a weak nodal domain must contain at least one vertex where $f$
is non-zero.  We will seek to understand the number of nodal domains
of $n$-th eigenfunction of the Schr\"odinger operator \eqref{eq:SO}.

\begin{definition}
 \label{def:nondeg_efun}
 We will call an eigenfunction \emph{non-degenerate} if it corresponds
 to a simple eigenvalue and its values at vertices are non-zero.
\end{definition}

In this paper only non-degenerate eigenfunctions $f$ will be
considered.\footnote{This behavior is generic with respect to
  perturbations of the potential $Q$.}  In such cases the strong and
weak nodal domains are the same and we denote by $\nu(G;f)$ their
number. Thus, any $f : G \rightarrow (\mathbb{R} \backslash 0
)^{V(G)}$ induces a bipartite $\nu(G;f)$-partition of $G$ with
components which are the nodal domains. The number of edges of $G$
which are deleted to generate the partition will be denoted by
$\zer(G;f)$, it is the number of sign flips of $f$ that occur along
edges of $G$. The number of independent cycles in $G$ where $f$ has a
constant sign (i.e. cycles that are contained in a single nodal
domain) will be denoted by $\ell(G;f)$. The following identity relates
these quantities \cite{BOS08}:
\begin{equation}\label{eq:IRU}
  \nu(G;f) = \zer(G;f) - \beta(G) + \ell(G;f) + 1 \  .
\end{equation}

We shall focus on partitions induced by eigenvectors $f^{(n)}$ of
(\ref{eq:SO}). Because of the special role played by this function, we
shall use the abbreviations $\nu_n=\nu(G;f^{(n)}),\
\zer_n=\zer(G;f^{(n)})\ {\rm and} \ \ell_n=\ell(G;f^{(n)})$.  Here,
Courant's theorem \cite{DGLS01}, and its extensions \cite{B} state
that the number of nodal domains is bounded by
\begin{equation}\label{eq:B}
  n - \beta(G) \le \nu_n \le n \ .
\end{equation}

Note that the lower bound on the number of nodal domains is not
optimal.  For well connected graphs such as $d$-regular graphs,
$\beta(G) =\frac{V(G)(d-2)}{2} + 1$ and hence, for $d\geq 4$, we have
$n-\beta(G)\le 0$ for all $n\le V(G)$.  An improved lower bound will be
derived below.

Following \cite{BBRS,HHT09}, we define,
\begin{definition}
  An eigenvector of the Hamiltonian (\ref{eq:SO}) is Courant sharp if
  $\nu_n=n$.
\end{definition}
The eigenvectors for the discrete Hamiltonian on trees are all Courant
sharp \cite{Fied} (see also \cite{B}).  While it is a special case of
(\ref{eq:B}) with $\beta(G)=0$, it is the first step in the proof of
the lower bound in \eqref{eq:B} for general $\beta$.

The chromatic number $\chi$ of a graph can be used to give a bound on
the number of nodal domains \cite{Oren07} : $\nu \le V-\chi+2$. Hence,
the only graphs for which the highest eigenvector (with $n=V(G)$) can
be Courant sharp, are bipartite graphs where $\chi=2$.

\section{Edge manipulations and the parameter dependent
  Hamiltonian} \label {section:preparations}

The main result of this paper is stated in Theorem
\ref{thm:morse}. Its formulation requires definitions, concepts and
facts which are provided in the first half of this section. To gain
some intuition to the more formal discussions, we start by making the
following simple observation.

Let $f=f^{(n)}$ be an non-degenerate eigenvector of the Hamiltonian
$H(G)$, see Definition~\ref{def:nondeg_efun}.
Let $G'$ be a connected factor of $G$ obtained by deleting an edge
$e=(i,j)$.  Following \cite{B}, we will modify the potential at the
vertices $i$ and $j$ in such a way that $f$ is an eigenfunction of the
factor $G'$ with an eigenvalue that equals $\lambda_n(G)$.  That is,
there is $m$, such that $\lambda_m(G') =\lambda_n(G)$.  Note that $n$
and $m$ are not necessarily equal, and the other eigenvalues and
eigenfunctions of $H(G')$ need not coincide with those of $H(G)$.  To
work out the necessary modification to the potentials we rewrite (\ref
{eq:SO}) at site $i$
\begin{equation*}
  (H(G) f)_i = -\sum_{k \sim i} f_k + q_i f_i
  = -\sum_{{k \sim i},\, {k\neq j}} f_k + \left(q_i - \frac{f_j}{f_i}\right) f_i
\end{equation*}
and similarly at site $j$,
\begin{equation*}
  (H(G) f)_j 
  = -\sum_{k \sim j,\,k\neq i} f_k + \left(q_j - \frac{f_i}{f_j}\right) f_j  \ .
\end{equation*}
Thus,
\begin{equation}
\label{eq:factorh} H(G) f =  \left(-A(G')  +
  \widetilde{Q}(G')\right) f\  ,
\end{equation}
where the potential $\widetilde{Q}(G')$ coincides with the
original potential $Q(G')$ on all vertices except the
vertices $i$ and $j$ where the potential is modified to be
\begin{equation}\label{eq:pot}
  \tilde q_i = q_i - \frac{f_j}{f_i} \qquad
  \tilde q_j = q_j - \frac{f_i}{f_j}.
\end{equation}
The operator which multiplies $f$ on the right hand side of
(\ref{eq:factorh}) above is a Hamiltonian operator for the factor
$G'$,
\begin{equation}
  H(G') =  -A(G')   + \widetilde{Q}(G')\  .
\end{equation}
Clearly $f$ is an eigenvector of $H(G')$ corresponding to an
eigenvalue which equals $\lambda_{n}(G)$ but whose position $m$ in
$\sigma(G')$ is not necessarily $n$.

This formal exercise acquires more substance once the modified
potentials are defined by replacing in (\ref {eq:pot}) the quotient
$\frac{f_j}{f_i}$ with a real non zero parameter $\alpha$ and $\frac{f_i}{f_j}$ with $1/\alpha$.  The resulting
parameter dependent Hamiltonian $H(G';\alpha)$ for the factor
$G'$ is defined as
\begin{equation}
  \label{eq:def_H_param}
  H(G';\alpha) = H(G) + B(\alpha).
\end{equation}
where the perturbation $B=B(\alpha)$ has only four non-vanishing entries,
\begin{align}
  \label{eq:B_r1}
  B_{i,i} &= -\alpha, &
  B_{i,j} &= 1 \\
  \label{eq:B_r2}
  B_{j,i} &= 1 &
  B_{j,j} &= -1/\alpha
\end{align}
The matrix $B(ij)$ has rank $1$, with a single non vanishing
eigenvalue $-\alpha-\frac{1}{\alpha}$.

The following theorem is due to Weyl (see, for example, \cite{Weyl})
and will be extensively used in the sequel.
\begin{theorem} \label{tm:weyl}
  Let $A, B, C$ be real and symmetric
  $N\times N$ matrices with $C=A+B$. Denote by $a_1\le a_2\le \cdots
  a_N$ and $c_1\le c_2\le \cdots c_N$ the spectra of $A$ and $C$
  respectively. If $B$ is of rank $1$ then the spectra of $C$ and $A$
  interlace in the following way:
 \begin{equation}
   \label{eq:Weylpos}
   \cdots \le a_{k-1}\le c_{k-1} \le a_{k } \le c_{k}\le \cdots \qquad
   \text{if } B\geq 0,
 \end{equation}
 and
 \begin{equation}
   \label{eq:Weylneg}
   \cdots \le a_{k-1}\le c_{k} \le a_{k } \le c_{k+1}\le \cdots \qquad
   \text{if } B\leq 0,
 \end{equation}
\end{theorem}

 After all these preparations we can arrive to the following theorem.
 \begin{theorem}
   \label{tm:critical}
   Let $G$ be a simple, connected graph with a Hamiltonian $H(G)$, and
   let $G'$ be a connected factor obtained by deleting the edge
   $(i,j)$ with parametrized Hamiltonian $H(G';\alpha)$ defined by
   \eqref{eq:def_H_param}-\eqref{eq:B_r2}.  Let
   $\lambda_m(G';\alpha)$ and $f^{(m)}(\alpha)$ be the $m$'th
   eigenvalue and the corresponding eigenvector of $H(G';\alpha)$.
   Consider the the eigenvalue $\lambda_m(G';\alpha)$ as a function of
   $\alpha$.  Its critical points $\alpha^{+}$ and $\alpha^{-}$ satisfy
   the equation
   \begin{equation}
     \label{eq:critpoint}
     \alpha = \pm \frac{f_j^{(m)}(\alpha)}{f_i^{(m)}(\alpha)}
  \end{equation}
   with the corresponding sign.

   Furthermore, as long as $\alpha^{+}\ne 0$ and finite,
   $\lambda_m(G';\alpha^{+})$ is an eigenvalue of $H(G)$ whose
   position in $\sigma(G)$ is $n\in \{m-1, m, m+1\}$.  The
   eigenvector $f^{(m)}(\alpha^{+})$ is the corresponding eigenvector
   of $H(G)$.

   Conversely, if $f^{(n)}$ is an non-degenerate eigenvector of $H(G)$
   then
     \begin{equation}
       \label{eq:crit_alpha}
       \alpha^+ = \frac{f^{(n)}_j}{f^{(n)}_i}
     \end{equation}
     is a critical point of $\lambda_m(G';\alpha)$ where $m$ is
     such that $n\in \{m-1, m, m+1\}$.
\end{theorem}

\begin{proof}
  If $f^{(m)}(\alpha)$ is normalized, by
  first order perturbation theory we have
  \begin{equation}
    \frac{{\rm d} \lambda_{m}(G',\alpha)}{{\rm d}\alpha}
    =\left( f^{(m)}(\alpha),\ \frac{{\rm d} H(G';\alpha)}
      {{\rm d}\alpha} f^{(m)}(\alpha) \right)
  \end{equation}
  The only $\alpha$ dependent entries of $H(G';\alpha)$ are the
  off-diagonal terms in $B$, see \eqref{eq:B_r1}-\eqref{eq:B_r2}.
  Hence
  \begin{equation}\label{eq:crit1}
    \frac{{\rm d} \lambda_{m}(G',\alpha)}{{\rm d}\alpha}
    = -\left(f^{(m)}_i(\alpha)\right)^2
    + \frac{1}{\alpha^2}\left(f^{(m)}_j(\alpha)\right)^2\ ,
  \end{equation}
  which is equal to zero at the solutions of
  equation~\eqref{eq:critpoint}.

  From now on, we consider only the \emph{relevant} critical points
  $\alpha^{+}$.  At $\alpha^{+}$ one can easily check that
  \begin{equation*}
    B(\alpha^+) f^{(m)}(\alpha^{+}) = 0
  \end{equation*}
  and therefore
  \begin{equation*}
    \lambda_m(G';\alpha^{+}) f^{(m)}(\alpha^{+})
    = H(G';\alpha^+) f^{(m)}(\alpha^{+})
    = H(G) f^{(m)}(\alpha^{+}).
  \end{equation*}
  Conversely, if $f^{(n)}$ is an non-degenerate eigenvector of $H(G)$
  then it was shown in equation \eqref{eq:factorh} that $f^{(n)}$ is an
  eigenvector of $H(G';\alpha)$ for $\alpha$ determined by equation
  \eqref{eq:crit_alpha}.

  Finally, the matrix $B$ is of rank-$1$ and Weyl's theorem \ref{tm:weyl}
  guarantees that the spectra of $H(G)$ and $H(G')$ interlace according
  to (\ref{eq:Weylpos}) or (\ref{eq:Weylneg}), showing that $n\in
  \{m-1, m, m+1\}$.
\end{proof}

Further insight can be gained by examining closely how the position of
the eigenvalue in the spectrum and the number of nodal domains change
depending on the sign and the type (minimum or maximum) of the
critical point.  Note that we will shift the point of view and will
use $G'$ as the base graph, i.e.\ we will investigate the change in
the above quantities when an edge is added back.

Let $\alpha^c$ denote a relevant critical point of
$\lambda_m(G';\alpha)$.  We start with the case $\alpha^c>0$.
Weyl's theorem (\ref{eq:Weylpos}) implies that for all $\alpha>0$
(that is $B\leq0$),
\begin{equation*}
  \lambda_{m-1}(G)\le \lambda_m(G';\alpha)\le \lambda_m(G).
\end{equation*}
At the critical point $\alpha=\alpha^c$ the value
$\lambda_m(G';\alpha^c)$ belongs to the spectrum of $H(G)$ and
therefore equals either $\lambda_{m-1}(G)$ or $\lambda_m(G)$.  Hence,
$\lambda_m(G';\alpha^{c})$ is a \emph{maximal} value of
$\lambda_m(G';\alpha)$ if it equals $\lambda_m(G)$ and a
\emph{minimal} value if it equals $\lambda_{m-1}(G)$.  For future use
it is convenient to introduce the following notation.  Let $M(ij)$
take the value $+1$ if the critical point is a maximum, and $0$ if it
is a minimum.  Let $\Delta n(ij)$ stand for the shift in the position
of the eigenvalue in the spectrum, from its position in $\sigma(G',
\alpha^c)$ to its position in $\sigma(G)$. One can summarize the
findings so far by the following statement:
\begin{equation}
  \label{eq:alphagt0}
  {\rm For} \  \alpha^c > 0 \   \ :\ \  M(ij) - \Delta n(ij) = 1.
\end{equation}
Similarly, for $\alpha^c<0 $, Weyl's theorem (\ref {eq:Weylneg}) implies
that
\begin{equation*}
  \lambda_{m}(G)\le \lambda_m(G';\alpha)\le \lambda_{m+1}(G).
\end{equation*}
Hence, $\lambda_m(G';\alpha^{c})$ must attain either the maximal value
$ \lambda_{m+1}(G)$ or the minimal value $\lambda_{m}(G)$.  Using the
same notation as above we find:
\begin{equation}
  \label{eq:alphalt0}
  {\rm For} \  \alpha^c < 0 \   \ :\ \  M(ij) - \Delta n(ij) = 0
\end{equation}

More information on the transition from $G$ to $G'$ is gained by viewing the eigenvector
$f^{(m)}(\alpha^{c})$ first as an eigenvector on $G'$ and then as an
eigenvector on $G$.  This dual view point is now applied to follow the
variation in the number of cycles with constant sign and the number of
nodal domains, as one counts them with respect to $G'$ or to $G$.  It
is important to remember that the sign of $\alpha^c$ is the relative
sign of the $i$ and $j$ components $f^{(m)}(\alpha^{c})$ across the
edge $(i,j)$.

Starting with $\alpha^{c}<0$, the number of nodal domains and the
number of loops of constant sign are not affected by the transition
from $G'$ to $G$ since $f_j / f_i <0$ implies that $i$ and $j$ are in
nodal domains with different signs.  Denoting by $\Delta \ell (ij)$
the change in the number of loops with constant sign, and by $\Delta
\nu (ij)$ the change in the number of nodal domain, we have
\begin{equation}
  \label{eq:alphalt0f}
  {\rm For} \  \alpha^c < 0 \   \ :\ \  \Delta \ell (ij) - \Delta \nu (ij) = 0\ .
\end{equation}

On the other hand,  for $\alpha^{c}>0$  either the number of nodal domains or the number of
cycles of constant sign will change.  Indeed, the edge $(i,j)$ either
connects two vertices that already belong to the same nodal domain,
increasing $\ell(f^{(m)})$ by $1$ (and leaving
$\nu(f^{(m)})$ unchanged) or it connects two nodal domains of the same
sign, in which case $\nu(f^{(m)})$ \emph{decreases} by $1$ while
$\ell(f^{(m)})$ remains constant.  This  leads to
\begin{equation}
\label{eq:alphagt0f}
{\rm For} \  \alpha^c > 0 \   \ :\ \  \Delta \ell (ij) - \Delta \nu (ij) = 1\ .
\end{equation}

Comparing (\ref{eq:alphalt0}) to (\ref{eq:alphalt0f}) and (\ref{eq:alphagt0}) to (\ref{eq:alphagt0f}), the observations above can summarized by:
\begin{equation} \label{eq:morse1}
\Delta \ell (ij)-\Delta \nu (ij)=  M(ij) -\Delta n (ij)\ ,
\end{equation}
which is valid for both signs of $\alpha^{c}$.

The discussion so far centered on a factor $G'$ which differs from the
original graph $G$ by the deletion of a single edge. However, by
successive applications of Theorem \ref{tm:critical} it can be
generalized to any connected factor of $G$ obtained by the elimination
of an arbitrary number of edges while modifying the appropriate vertex
potentials. The set of parameters will be denoted by $\alpha =\{
\alpha_e : e \in \mathcal{E}(G)\setminus \mathcal{E}(G') \}$, the
parameter dependent Hamiltonian is $H(G',\alpha)$, and
$\lambda_m(G';\alpha)$ is its $m$'th eigenvalue.

\begin{corollary}
  \label{cor:thm1}
  Let $G$ be a graph as previously and $G'$ a connected factor
  obtained by deleting $k \le \beta(G)$ edges from $G$. Let $\alpha^c$
  be a relevant critical point of $\lambda_m(G';\alpha)$. Then,
  provided that none of the components of $f^{(m)}(\alpha^c)$ vanishes
  at vertices where edges were added,
  $\lambda_m(G';\alpha^c)=\lambda_n(G)\in \sigma(G)$ with $|n-m|\le
  k$.  The corresponding eigenvectors are the same.
\end{corollary}

The lower bound on the number of nodal domains, the left part of
equation~\eqref{eq:B}, was proved by chaining the operations of edge
removal \cite{B}.  In fact, a more careful book-keeping allows us to sharpen
the lower bound (see \cite{BBRS} for the same inequality in the
context of quantum graphs).

\begin{theorem}
  \label{thm:nodal_bound_mod}
  Let $f^{(n)}$ be the $n$-th eigenfunction of the Hamiltonian $H(G)$
  such that the corresponding eigenvalue is simple and $f^{(n)}$ has
  no zero components.  Then the number of nodal domains of $G$ with
  respect to $f^{(n)}$ satisfies
  \begin{equation}
    \label{eq:nodal_bound_mod}
    n-(\beta(G) - \ell_n) \leq \nu_n \leq n.
  \end{equation}
  Correspondingly, the number $\zer_n$ of edges across which the
  eigenvector $f^{(n)}$ changes its sign satisfies the bound
  \begin{equation}
    \label{eq:zero_bound_mod}
    n  - 1 \le \zer_n \le n +(\beta(G)-\ell_n) - 1.
  \end{equation}
\end{theorem}

\begin{remark}
  Note that the lower bound on $\nu_n$, the quantity $n-(\beta(G) -
  \ell_n)$ is always non-negative.
\end{remark}

\begin{remark}
Equation (\ref{eq:zero_bound_mod}) is the generalization of Sturm's  oscillation theorem to discrete graphs.
\end{remark}

\begin{proof}[Proof of Theorem~\ref{thm:nodal_bound_mod}]
  We cut $\beta$ edges of the graph $G$, modifying the potential
  accordingly, until we arrive to a tree $T$, such that $f=f^{(n)}$ is
  it eigenfunction.  It is eigenfunction number $m$ and, since it is a
  tree, Fiedler theorem \cite{Fied} implies it has $\nu_m(T) = m$
  nodal domains.  Also, $\ell_m(T) = 0$, because there are no cycles
  on a tree.

  We rewrite equation~\eqref{eq:morse1} in the form
  \begin{equation*}
    \Delta \nu(e) = \Delta n(e) + \Delta\ell(e) - M(e),
  \end{equation*}
  where $e$ is the removed edge.  Adding back the removed edges one by
  one and adding the above identities to the equation $\nu_m(T) = m$
  we arrive at
  \begin{equation*}
    \nu_n(G) = n + \ell_n - \sum M(e).
  \end{equation*}
  Since the number of maxima in the sequence is at most $\beta$, the
  number of nodal domains $\nu_n(G)$ is at least $n+\ell_n-\beta$,
  proving inequality~\eqref{eq:nodal_bound_mod} (the upper bound is due to
  \cite{DGLS01}).  Substituting equation~\eqref{eq:IRU} for $\nu_n$,
  we obtain inequality~\eqref{eq:zero_bound_mod}.
\end{proof}

\begin{corollary}
  \label{cor:tree_partition_sharp}
  If $f^{(n)}$ satisfies the conditions of
  Theorem~\ref{thm:nodal_bound_mod} and its nodal partition graph
  $\tG$ is a tree (or, equivalently, the edges on which $f^{(n)}$
  changes sign do not lie on cycles of the graph), then $f^{(n)}$ is
  Courant-sharp: $\nu_n=n$.
\end{corollary}

\begin{proof}
  Indeed, if $\tG$ is a tree, then no cycles are broken when removing
  the edges connecting the nodal domains and $\ell_n=\beta$.
\end{proof}

\section{Critical Partitions -  the main theorems  and  proofs} \label {section:results}

So far we discussed the reduction of a graph to its connected
factors. The generation of partitions (disconnected factors) requires
the introduction of some more concepts and definitions. Let
$\mathcal{P}$ be a bipartite $\nu(G)$-partition of $G$, see
Definition~\ref{def:partition} and the discussion below it.  We denote
by $\mathcal{E(P)}$ the edge set of $\cP$.  Since we wish to
characterize the partitions that appear as nodal domains, we will
ofter refer to the connected components $P_k$ of $\cP$ as ``domains''.

Let $\zer(G;{\mathcal P})$ denote the number of edges in
$\mathcal{E}(G) \backslash \mathcal{E}(\mathcal{P})$. Construct the
Hamiltonian $H(\mathcal{P};\alpha)$, $\alpha \in \mathbb{R}^{\zer(G)}$
where for each deleted edge $e=(i,j)$, $i<j$, the potentials
$\alpha_{e}$ and $\alpha_{e}^{-1}$ are added to the vertices $i$ and
$j$ as was explained above.  Note that the ordering of the vertices in
an edge can be chosen arbitrarily; we chose $i<j$ for definiteness.
The Hamiltonian $H(\mathcal{P};\alpha)$ is block-diagonal, namely
\begin{equation*}
  H(\mathcal{P};\alpha)=\bigoplus_{k=1}^{\nu(G)}H(P_k;\alpha).
\end{equation*}
Let $\lambda_1(P_k;\alpha)$ be the lowest eigenvalue of
$H(P_k;\alpha)$ with the corresponding eigenvector denoted by
$g(P_k;\alpha)$.  Note that since $g(P_k;\alpha)$ is a ground state,
all its components have the same sign.  We extend $g(P_k;\alpha)$ to
all vertices of the graph $G$ by setting $g_s(P_k;\alpha)=0$ for $s
\not\in \cV(P_k)$.  We assume $g(P_k;\alpha)$ is normalized and
non-negative,
\begin{equation}
  \label{eq:normalize}
  \sum_{s\in \cV(P_k)} \left|g_s(P_k;\alpha)\right|^{2}
  = \sum_{s\in \cV(G)} \left|g_s(P_k;\alpha)\right|^{2}
  = 1,
  \qquad g_s(P_k;\alpha) \geq 0.
\end{equation}

\begin{definition}
  \label{def:equi}
  Let $\cP$ be a bipartite $\nu$-partition of $G$ with the parameter
  dependent Hamiltonian $H(\mathcal{P};\alpha)$.  An
  \emph{equipartition} is a pair $(\cP, \alpha')$, where $\alpha' \in
  \mathbb{R}^{\zer(G;{\mathcal P})}$ is a vector of parameters such
  that all the lowest eigenvalues $\lambda_1(P_k;\alpha')$ are equal,
  \begin{equation} \label{eq:equip}
    \lambda_1(P_1;\alpha') = \lambda_1(P_2;\alpha')
    = \cdots = \lambda_1(P_{\nu};\alpha').
  \end{equation}
\end{definition}
Consider the set $\cQ(\mathcal{P})\subset
\mathbb{R}^{\zer(G;{\mathcal P})}$ in the space of parameters $\alpha$
where equation~\eqref{eq:equip} is satisfied (equivalently, all
$\alpha$ such that $(\cP, \alpha)$ is an equipartition).  On this set,
the function $\Lambda(\mathcal{P};\alpha')= \lambda_1(P_k;\alpha')$
will be referred to as the equipartition energy.  Obviously, the index
$k$ is arbitrary.

Intuitively, the $(\nu-1)$ equalities in (\ref {eq:equip}) reduce
the number of independent variables required to define the
equipartition energy from $\zer(G;{\mathcal P})$ to $\eta :=
\zer(G;{\mathcal P})-(\nu-1)$.  Notice that, because of
(\ref{eq:IRU}),
\begin{equation}
\eta = \zer(G;\mathcal{P})-(\nu-1) = \beta(G) -\ell(G;\mathcal{P}) \ge 0 \ .
\end{equation}

If the partition $\mathcal{P}$ is induced by an eigenvector $f^{(n)}$
of $H(G)$ then it also generates an equipartition with the vector
$\alpha^c$ of parameters defined by
\begin{equation}
  \label{eq:induced_equipartition}
  \alpha^c_e=\frac{f^{(n)}_j}{f^{(n)}_i},
  \qquad \text{for each } e=(i,j) \in
  \cE(G) \setminus \cE(\cP),
  \quad i<j.
\end{equation}
Then the equipartition energy coincides with the eigenvalue of $G$,
$\Lambda(\mathcal{P};\alpha^c)= \lambda_n(G)$.  In the sequel we shall
consider the equipartition energy in the neighborhoods of these
special points.  Note that in the definition above the vertices $i$
and $j$ belong to different nodal domains.  Therefore the special
values of the parameters $\alpha_e$ are negative.  This is why from
now on we restrict our attention to the negative subspace of the
parameter space, $\R_-^{\zer(G,\cP)}$.

The main results of the paper can now be formulated in terms of two
theorems which distinguish between the cases $\eta=0$ and $\eta>0$.

\begin{theorem}
  \label{thm:tree}
  Let $\mathcal{P}$ be a bipartite $\nu$-partition of $G$ with $\eta =
  \zer(G;\mathcal{P})-(\nu-1) = 0$.  Then $\cP$ is a tree partition
  and there exists at most one equipartition $(\cP, \gamma)$, such that
  $\gamma \in \R_-^{\zer(G; \cP)}$.  The value $\lambda
  =\Lambda(\mathcal{P};\gamma)$ is the $\nu$-th eigenvalue in the
  spectrum of $G$, and the corresponding eigenvector is Courant sharp.
\end{theorem}

\begin{proof}
  By definition of the partition graph $\tG(\cP)$, it has $\nu$
  vertices and $\zer(G; \cP)$ edges.  Therefore its Betti number is
  $\zer(G;\mathcal{P})-\nu+1 = 0$ and it is a tree (see an
  illustration in Fig.~\ref{fig:partition}(b) ignoring the dashed edges).

  Let $\gamma$ be a vector of negative parameters such that $(\cP,
  \gamma)$ is an equipartition.
  The functions $g(P_k,\gamma)$, normalized as in (\ref{eq:normalize}),
  can be used to construct an eigenvector of $H(G)$.  Consider
  \begin{equation}
    \label{eq:construct_tree_eigenstate}
    f=\sum_{q=1}^{\nu_n} t_k \, g(P_k;\gamma),
  \end{equation}
  where the coefficients $t_k$ are determined as follows. Pick
  arbitrarily one of the domain in $\mathcal{P}$, say $P_1$, to be the
  ``root'' of the partition tree and let $t_1=1$.  For every domain
  $P_r$ adjacent to $P_1$ on the partition tree, let $e=(i,j)$ be the
  unique edge connecting $P_1$ to $P_r$ with $i\in P_1$ and $j\in P_r$.
  Without loss of generality assume that $i<j$ and take
  \begin{equation*}
    t_r = t_1 \gamma_{e} \frac{g_i}{g_j}.
  \end{equation*}
  Continue the process recursively to domains at larger distances from
  $P_1$. The tree structure guarantees that this covers the entire set
  of domains on the graph without ever reaching a domain for which the
  variable $t$ has been previously computed.  Note that the values of
  $t$ alternate in sign, leading to the function $f$ whose nodal
  partition is precisely $\cP$.  We will now show that the resulting
  function $f$, equation~\eqref{eq:construct_tree_eigenstate}, is an
  eigenvector of $H(G)$.

  Start with the Hamiltonian $H(\cP; \gamma)$ and observe that since
  $f$ is composed out of eigenvectors of the connected components of
  $H(\cP; \gamma)$ that \emph{share the same eigenvalue},
  \begin{equation*}
    H(\cP; \gamma) f = \Lambda(\mathcal{P};\gamma) f.
  \end{equation*}
  We will now add the edges from $\cE(G) \setminus \cE(\cP)$ one by
  one.  For simplicity of notation we will consider the addition of
  the edge $e=(i,j)$ between the components $P_1$ and $P_r$.  The
  addition of the edge results in a matrix $B$ defined by
  \eqref{eq:B_r1}-\eqref{eq:B_r2} with $\alpha=\gamma_e$.  The
  relevant elements of the vector $f$ are
  \begin{equation*}
    f_i = t_1 g_1, \qquad f_j = t_r g_j = t_1 \gamma_e g_i.
  \end{equation*}
  A direct computation shows that $Bf = 0$ and therefore $f$ remains an
  eigenfunction of the modified graph with the same eigenvalues.
  Proceeding similarly with the other edges we conclude that $H(G)f =
  \Lambda(\mathcal{P};\gamma) f$.

  The function $f$ has $\nu$ nodal domains with respect to the graph
  $G$.  From Corollary~\ref{cor:tree_partition_sharp} we get that the
  function $f$ as an eigenfunction of $H(G)$ is Courant-sharp, that
  is, its index is $\nu$.
  Finally, since there is at most one Courant-sharp eigenfunction with
  $\nu$ domains, the equipartition is unique.
\end{proof}

Let $f^{(n)}$ be a non-degenerate eigenvector corresponding to the
eigenvalue $\lambda_n(G)$.  If the corresponding partition $\cP$ is
not a tree, there will be other equipartitions locally around the
special point~\eqref{eq:induced_equipartition}. Namely, we shall consider
 a ball $\mathcal{B}_{\epsilon}(\alpha^{c})$
of radius $\epsilon$ centered at $\alpha^{c}$ defined by
\eqref{eq:induced_equipartition}.  The value of $\epsilon>0$ is chosen
sufficiently small so that the following two conditions are satisfied:
(\emph{i}) the variation in the equipartition energy
$\Lambda(\mathcal{P};\alpha)$ is smaller than the minimal separation
between successive eigenvalues in $\sigma(G)$ and  (\emph{ii}) none of
the hyper-planes $\alpha_{e}=0$ intersect the ball. The discussion
which follows is restricted to $\alpha\in
\mathcal{B}_{\epsilon}(\alpha^{c})$ and thus, only local properties of
the equipartition energy function are considered.

\begin{theorem}
  \label{thm:morse}
  Let $\mathcal{P}$ be a bipartite $\nu$ partition of $G$ with $\eta =
  \zer(G;\mathcal{P})-(\nu-1) > 0$ and $H(\mathcal{P};\alpha)$ be the
  associated parameter dependent Hamiltonian.  If there exists an
  equipartition $(\cP, \ta)$, then, in the vicinity of $\ta$ the
  set $\cQ(\mathcal{P})$ of equipartitions
  forms a smooth $\eta$-dimensional submanifold of $\mathbb{R}^{\zer(G)}$.

  The point $\alpha^c\in \R_-^{\zer(G)}$ is a critical point of the
  equipartition energy $\Lambda(\mathcal{P};\alpha)$ on the manifold
  $\cQ(\mathcal{P})$ if and only if it corresponds to an
  non-degenerate eigenfunction with eigenvalue
  $\Lambda(\mathcal{P};\alpha^c)$ and nodal domains $\cP$.
  Furthermore, if $\Lambda(\mathcal{P};\alpha^c)$ is the $n$-th
  eigenvalue in the spectrum, the Morse index $\morse_n$ of the
  critical point $\alpha^c$ satisfies
  \begin{equation} \label{eq:morse}
    \morse_n \ =\ n-\nu_n .
  \end{equation}
\end{theorem}

\begin{proof}
  We start by describing a parametrization of the manifold of
  equipartitions by $\eta$ independent parameters
  $\{\xi_k\}_{q=1}^{\eta}$.  The construction is local around an
  existing equipartition.  Start with the partition $\mathcal{P}$ (see
  Fig.~\ref{fig:partition}(a)), and remove $\eta = \zer-(\nu -1)$
  edges which connect different nodal domains, leaving $\nu(G) -1$
  bridge edges which turn the partition graph $\tilde G(\mathcal{P})$
  into a tree (as in Fig.~\ref{fig:partition}(b)).  The set of removed
  edges will be denoted by $X(\cP)$.  Every deletion of an edge
  $e=(i,j)$ is accompanied by the modification of the vertex
  potentials at the vertices $i$ and $j$ by a parameter $\xi_{e}$ as
  usual.  This partial set of $\eta$ parameters $\xi_e$ uniquely
  defines the equipartition energy. Indeed, the graph which was
  produced by the deletion of edges and which we denote by $G_{\xi}$
  is exactly of the kind which was discussed in theorem
  \ref{thm:tree}.  Since we started at an existing equipartition
  $(\cP,\ta)$, it corresponds to an equipartition $(\cP,
  \widetilde{\beta})$ of the graph $G_{\widetilde{\xi}}$.  Here
  $\widetilde{\beta}$ is a set of $\nu(G) -1$ entries of $\ta$ that
  correspond to edges $e \not\in X(\cP)$.  The $\eta$ parameters
  $\widetilde{\xi}$, on the other hand, are the entries of $\ta$ that
  correspond to the edges $e \in X(\cP)$.

  From Theorem~\ref{thm:tree} we conclude that the equipartition
  corresponds to a Courant-sharp eigenfunction
  $\psi^{(\nu)}(G_{\widetilde\xi})$.  Under local variation of the
  parameters $\xi$ the eigenfunction $\psi^{(\nu)}(G_{\xi})$ remains
  Courant-sharp and induces an equipartition of the graph $G_{\xi}$
  and, therefore, of the graph $G$.  Moreover, the eigenfunction is
  smooth as a function of the parameters $\xi$.  This allows us to
  define the parameters $\alpha$ (thus constructing a locally smooth
  immersion $\cQ(\mathcal{P}) \to \mathbb{R}^{\zer(G)}$ ) in the
  following manner.  There is one parameter for each edge in $\cE(G)
  \setminus \cE(\cP)$.  If the edge $e$ is in the set $X(\cP)$, we
  take $\alpha_e = \xi_e$.  Otherwise, we compute $\alpha_e$ from the
  eigenfunction $\psi^{(\nu)}(G_{\xi})$ according to the familiar
  prescription, see equation~\eqref{eq:crit_alpha}.


  \begin{figure}
    \label{fig:partition}
    \vspace{1mm}
    {\includegraphics[scale=0.5]{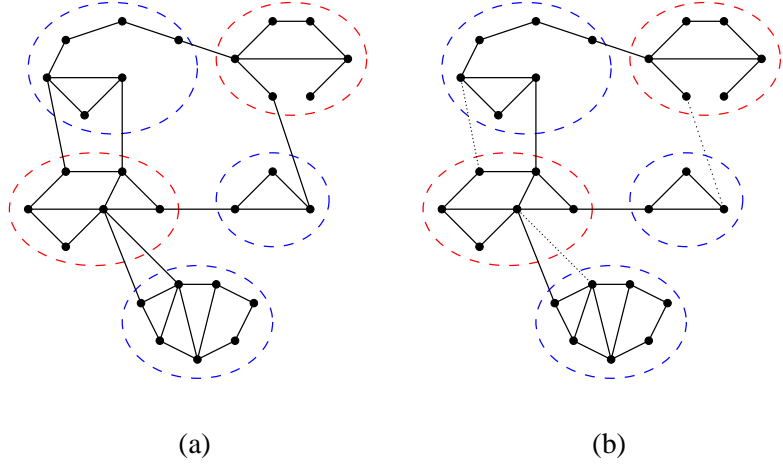}}
    \caption{(a) A graph with a bipartite partition. Circles enclose
      vertices within the same domain.(b) The dashed edges are removed
      from (a) to generate a tree-partition.}
  \end{figure}


  We now prove that the critical points of the equipartition energy
  $\Lambda(\mathcal{P};\alpha)$ correspond to the eigenvectors of
  $H(G)$ and vice versa.  The search for critical points should be
  carried out in the manifold $\cQ(\mathcal{P})$.  To do so we search
  for the critical points of $\lambda_1(P_1;\alpha)$ and impose the
  restriction to $\cQ(\mathcal{P})$ by introducing $(\nu-1)$
  Lagrange multipliers $\{ \zeta_k\}$.  One has to search for the
  critical points of
  \begin{equation}
    \tilde \Lambda (\mathcal{P};\alpha) = \lambda_1(P_1;\alpha)
    + \sum_{k=1}^{\nu-1} \zeta_k \left( \lambda_1(P_k;\alpha)
      - \lambda_1(P_1;\alpha) \right),
  \end{equation}
  The sum can be written in a more concise way as
  \begin{equation}
    \label{eq:concise_sum}
    \tilde \Lambda (\mathcal{P};\alpha)
    = \sum_{k=1}^{\nu} c_k \lambda_1(P_k;\alpha),
  \end{equation}
  where the $c_k$ are linear in the $\zeta_k$.  For every $e=(i,j)$
  only two terms in \eqref{eq:concise_sum} depend on the parameter
  $\alpha_e$, namely the terms corresponding to $P_{\chi(i)}$ and
  $P_{\chi(j)}$, where the function $\chi(\cdot)$ maps a vertex to the
  number of the corresponding domain.  Taking the derivative with
  respect to $\alpha_{ij}$, we get
  \begin{equation}
    \frac{\partial\tilde\Lambda(\mathcal{P};\alpha)}{\partial \alpha_{ij}}
    = c_{\chi(i)} \frac{\partial \lambda_1 (P_{\chi(i)};\alpha)}{\partial \alpha_{ij}}
    + c_{\chi(j)} \frac{\partial \lambda_1 (P_{\chi(j)};\alpha)}{\partial \alpha_{ij}}.
  \end{equation}
  As previously, let $g(P_k;\alpha)$ be the first eigenvector of the
  domain Hamiltonian $H(P_k;\alpha)$, normalized and positive.  Using
  first-order perturbation theory and the explicit dependence of the
  potential on $\alpha_{ij}$, similar to what is done in theorem
  \ref{tm:critical} the following equations must be satisfied at the
  critical point for every $(i,j)\in \mathcal{E}(G) \setminus
  \mathcal{E(P)}$:
  \begin{equation}
    \label{eq:crit_condition}
    c_{\chi(i)} g_i(P_{\chi(i)};\alpha)^2
    - c_{\chi(j)} \frac{1}{\alpha_{ij}^2} g_j(P_{\chi(j)};\alpha)^2 = 0.
  \end{equation}
  We immediately conclude that the critical values of the Lagrange
  multipliers $c_\chi$ are non-negative.  Form the function
  \begin{equation}
    \label{eq:efun_from_crit}
    f = \sum_{k=1}^{\nu-1} \pm\sqrt{c_k} g(P_k;\alpha),
  \end{equation}
  where the signs are to be chosen in accordance with the bipartite
  structure of the partition $\cP$.  Then
  equation~\eqref{eq:crit_condition} can be written as
  \begin{equation*}
    f_i^2 - \frac1{\alpha_{ij}^2} f_j^2 = 0
    \qquad \text{or} \qquad
    \alpha_{ij} = \frac{f_j}{f_i},
  \end{equation*}
  which describes the values of $\alpha_e$ at the critical points.
  The function $f$ is an eigenfunction of the Hamiltonian $H(\cP, \alpha^c)$
  --- in fact one belonging to the $(\nu-1)$-dimensional eigenspace of
  the degenerate lowest eigenvalue $\Lambda(\cP;\alpha^c)$.  It can
  now be shown by explicit calculation, see the discussion
  leading up to equation~\eqref{eq:factorh}, that $f$ is also an
  eigenfunction of the original Hamiltonian $H(G)$.

  Conversely, starting from an eigenfunction $f$ with the given nodal
  domain structure $\cP$, we define
  \begin{equation*}
    \alpha_{ij} = \frac{f_j}{f_i}, \qquad
    c_k = \sum_{v\in P_k} f_v^2,
  \end{equation*}
  and check that the condition for the critical point~\eqref{eq:crit_condition} is satisfied.

  Finally, to show~\eqref{eq:morse}, we go back to the graph $G_\xi$
  and start the process of successive addition of edges and use
  \eqref{eq:alphalt0} at each step, since all parameters $\alpha_{e}$
  are negative.  Thus, the change in the position in the spectrum
  $\Delta n$ over $\eta$ additions equals the number of times the
  critical point is approached as a maximum:
  \begin{equation*}
    \Delta n =  \sum M (e).
  \end{equation*}
  Since at the beginning of the process we had a Courant-sharp
  eigenfunction with $\nu$ domains, its position in the spectrum was
  $\nu$.  Thus the change $\Delta n$ of position is from $\nu$ to $n$,
  or,
 \begin{equation}
    n-\nu_n = \morse_n,
  \end{equation}
  where $\morse_n$ is the total number of independent directions in
  which one approaches the critical point of the equipartition as a
  maximum; the Morse index (see \cite[Section 4.3]{BBRS} for more detail).

  The above construction was carried out for a particular
  parametrization of the manifold in the space of parameters
  corresponding to equipartitions. However, since $\Lambda(P;\alpha)$
  is analytic in the neighborhood of the critical point, the Morse
  index does not depend on the choice of coordinates.
\end{proof}

\section{Acknowledgements}

The authors cordially thank Rami Band for many invaluable discussions
and comments. The work of HR and US was supported by an EPSRC grant
(EP/G02187), The Israel Science Foundation (ISF) grant (169/09) and
from the Welsh Institute of Mathematical and Computational Science
(WIMCS).  The work of GB was suported by the National Science Foundation under
Grant No. DMS-0907968.  HR and GB thank the Department of Complex
Systems at the Weizmann Institute for the hospitality extended during
their visits.

\bibliographystyle{abbrv}

\bibliography{partitions1}

\end{document}